\documentclass{article}

\usepackage[margin=1in]{geometry}
\usepackage{paralist}
\usepackage{amsmath}
\usepackage{amsthm}
\usepackage{latexsym}
\usepackage{url,hyperref}
\usepackage{boxedminipage}
\usepackage{array}
\usepackage{algorithm}
\usepackage{flafter}

\usepackage{amssymb}
\usepackage{color}
\usepackage{makeidx}
\usepackage[font={small}, margin=2em]{caption}

\usepackage{float}
\usepackage{graphicx}

\graphicspath{{graphics/}}

\usepackage{soul}

\newtheorem{theorem}{Theorem}
\newtheorem*{theorem*}{Theorem}

\newtheorem*{conjecture*}{Conjecture}

\newtheorem*{remark*}{Remark}
\newtheorem{lemma}[theorem]{Lemma}
\newtheorem*{lemma*}{Lemma}

\newtheorem*{proposition*}{Proposition}
\newtheorem{corollary}[theorem]{Corollary}
\newtheorem*{corollary*}{Corollary}

\newtheorem*{property*}{Property}

\DeclareMathOperator{\Binom}{Binom}

\DeclareMathOperator{\Span}{span}
\DeclareMathOperator{\OPT}{OPT}
\DeclareMathOperator{\even}{even}
\DeclareMathOperator{\odd}{odd}
\DeclareMathOperator{\argmin}{argmin}

\newcommand{\weight}{\mathrm{w}}

\newcommand{\calE}{\mathcal{E}}

\newcommand{\E}{{\bf E}}

\def\b1{{\bf 1}}

\author{
Patrick Jaillet%
\thanks{
Dept.~of Electrical Engineering and Computer Science, MIT,
Cambridge. E-mail: {\tt jaillet@mit.edu}.
Supported in part by NSF grant 1029603,
and by ONR grants N00014-12-1-0033 and N00014-09-1-0326.
}
\and
Jos\'e A.~Soto%
\thanks{
DIM-CMM University of Chile, Santiago and Technische Universit\"at Berlin, Berlin.
E-mail: {\tt jsoto@dim.uchile.cl}.
Supported in part by N\'ucleo Milenio Informaci\'on y
Coordinaci\'on en Redes ICM/FIC P10-024F.
}
\and
Rico Zenklusen%
\thanks{
ETH Zurich, Zurich, and
Johns Hopkins University, Baltimore.
E-mail: {\tt ricoz@math.ethz.ch}.
Supported by NSF grants CCF-1115849 and CCF-0829878,
and by ONR grants N00014-12-1-0033, N00014-11-1-0053
and N00014-09-1-0326.
}
}

\begin{document}

\title{Advances on Matroid Secretary Problems:\\
Free Order Model and Laminar Case\thanks{A short
version of this paper appeared at
IPCO 2013~\cite{jaillet_2013_advances}.}}

\maketitle

\setcounter{footnote}{0}

\begin{abstract}
The most important open conjecture in the context of
the matroid secretary problem claims the existence of
an $O(1)$-competitive algorithm applicable to any matroid.
Whereas this conjecture remains open, modified forms of
it have been shown to be true, when assuming that the
assignment of weights to the secretaries is not adversarial
but uniformly at
random~\cite{soto_2013_matroid,oveisgharan_2011_variants}.
However, so far, no variant of the matroid secretary
problem with adversarial weight assignment is known
that admits an $O(1)$-competitive algorithm.
We address this point by presenting a
$4$-competitive procedure for the \emph{free order model},
a model suggested shortly after the introduction
of the matroid secretary problem, and for which
no $O(1)$-competitive algorithm was known so far.
The free order model is a relaxed version of the original
matroid secretary problem, with the only difference 
that one can choose the order in which secretaries are
interviewed.

Furthermore, we consider the classical matroid secretary problem
for the special case of laminar matroids.
Only recently, an
$O(1)$-competitive algorithm has been found for this case,
using a clever but rather involved method
and analysis~\cite{im_2011_secretary} that leads
to a competitive ratio of $16000/3$.
This is arguably one of the most involved special cases of the
matroid secretary problem for which an $O(1)$-competitive
algorithm is known.
We present a considerably simpler and stronger
$3\sqrt{3}e\approx 14.12$-competitive procedure,
based on reducing the problem to a matroid secretary problem
on a partition matroid.
Furthermore, our procedure is order-oblivious, which,
as shown in~\cite{azar_2014_prophet}, allows for
transforming it into a
$3\sqrt{3}e$-competitive algorithm for
single-sample prophet inequalities.

\end{abstract}

\section{Introduction}

The secretary problem is a classical online selection problem
of unclear origin~\cite{dynkin_1963_optimum,ferguson_1989_who,%
gardner_1960_mathematical,gardner_1960b_mathematical,%
lindley_1961_dynamic}.
In its original form, the task is to choose the best
out of $n$ secretaries, also called \emph{elements}
or \emph{items}. Secretaries arrive (or are interviewed)
one by one in random order. As soon as a secretary arrives,
she can be ranked against all previously seen secretaries.
Then, before the next one arrives, one has to decide
irrevocably whether to choose the current secretary or not.
There is a classical algorithm that selects the best secretary
with probability $1/e$~\cite{dynkin_1963_optimum},
and this is known to be asymptotically optimal.
In its initial form, the secretary problem was
essentially
a stopping time problem, and not surprisingly, it mainly
attracted the interest of probabilists.

Recently, secretary problems enjoyed a revival, and
various generalizations
have been studied. These developments are strongly motivated by a
close connection to online
mechanism design, where a good is sold to agents arriving
online~\cite{kleinberg_2005_multiple-choice,babaioff_2007_matroids}.
Here, the agents correspond to the secretaries and they reveal
prices that they are willing to pay in exchange for goods.
This leads to secretary problems where more than
one secretary can be chosen.
The most canonical generalization asks to hire $k$
out of $n$ secretaries, each revealing a non-negative weight
upon arrival, and the goal is to hire a maximum weight subset
of $k$ secretaries. This interesting variant was introduced
and studied by Kleinberg~\cite{kleinberg_2005_multiple-choice},
 who presented a $(1-O(1/\sqrt{k}))$-competitive algorithm
for this setting.
However, in many applications, additional constraints
have to be imposed on the elements that can be chosen.
A very general class of constrained secretary problems,
where the chosen elements have to form an independent set of
a given matroid $M=(N,\mathcal{I})$, was
introduced by Babaioff, Immorlica and
Kleinberg~\cite{babaioff_2007_matroids}%
\footnote{A matroid $M=(N, \mathcal{I})$ consists of a finite
set $N$, called the \emph{ground set}, and a non-empty family
$\mathcal{I}\subseteq 2^N$ of subsets of $N$,
called \emph{independent sets}, satisfying:
(i) $I\in \mathcal{I}, J\subseteq I \Rightarrow J\in\mathcal{I}$,
and (ii) $I,J\in \mathcal{I}, |I|>|J|
\Rightarrow \exists f\in I\setminus J$ with $J\cup \{f\}\in \mathcal{I}$.
For more information on matroids we refer the reader
to~\cite{schrijver_2003_combinatorial}.}.
This setting, now generally termed \emph{matroid secretary problem},
covers at the same time many interesting cases and has a
rich structure that can be exploited to design 
algorithms with strong competitive ratios.

To give a concrete example of a matroid secretary problem,
and to motivate some of our results, consider the following
connection problem.
Given is an undirected graph $G=(V,E)$, representing
a communication network, with non-negative
edge-capacities $c:E\rightarrow \mathbb{Z}_{\geq 0}$ and a server
$r\in V$.
Clients, which are the equivalent of candidates
in the secretary problem, reside at vertices of the graph
and are interested in being connected to the server $r$ via
a unit-capacity path.
The number of clients and their locations are known.
Each client has a price that she is willing to pay to
connect to the server. These prices are unknown and no
assumptions are made on them except for being non-negative.
Clients then reveal themselves one by one in random order,
announcing their price.
Whenever a client reveals herself,
the network operator has to decide irrevocably before the
next client appears whether to serve this client
and receive the announced price.
The goal is to choose a maximum weight subset of clients
that can be served simultaneously without exceeding the
given capacities $c$.
It is well-known that the constraints imposed by the limited
capacity on the clients that can be chosen is a special type
of matroid constraint, namely a gammoid
constraint~\cite{schrijver_2003_combinatorial}.

For the classical matroid secretary problem, as discussed
above, the currently asymptotically best competitive algorithm
is an $O(\sqrt{\log{\rho}})$-competitive method by
Chakraborty and Lachish~\cite{chakraborty_2012_improved},
where
$\rho$ is the rank of the matroid%
%
%
.
This improved
on an earlier $O(\log{\rho})$-competitive algorithm of
Babaioff, Immorlica
and Kleinberg~\cite{babaioff_2007_matroids}.
Babaioff et al.~\cite{babaioff_2007_matroids} asked about
the existence of an $O(1)$-competitive algorithm
for the matroid secretary problem. This question remains
open and is arguably the currently most important open question
regarding the matroid secretary problem.

Motivated by this conjecture, many interesting advances have
been made to obtain $O(1)$-competitive methods, either
for special cases of the matroid secretary problem or variants
thereof.
In particular, $O(1)$-competitive algorithms have been found
for graphic matroids~\cite{babaioff_2007_matroids,korula_2009_algorithms}
(currently best competitive ratio: $2e$),
transversal matroids~\cite{babaioff_2007_matroids,%
dimitrov_2008_competitive,korula_2009_algorithms}
($8$-competitive),
co-graphic matroids~\cite{soto_2013_matroid} ($3e$-competitive),
linear matroids with at most $k$ non-zero entries
per column~\cite{soto_2013_matroid}
($k e$-competitive),
laminar matroids~\cite{im_2011_secretary}
($16000/3$-competitive),
regular matroids ($9e$-competitive)~\cite{dinitz_2013_matroid}, and
some types of decomposable matroids, including
max-flow min-cut matroids~\cite{dinitz_2013_matroid}
($9e$-competitive).
For most of the above special cases, strong competitive algorithms
have been found, typically based on very elegant techniques.
However for the laminar matroid, only a considerably
higher competitive ratio is known due to
Im and Wang~\cite{im_2011_secretary},
using a very clever but quite involved method and analysis.

Furthermore, variants of the matroid secretary problem have
been investigated that assume random instead of adversarial
assignment of the weights, and for which $O(1)$-competitive
algorithms can be obtained without any restriction on the
underlying matroid.
Recall that the classical matroid secretary problem
does not make any assumptions on how
weights are assigned to the
elements, which means that we have to assume a worst-case,
i.e., \emph{adversarial}, weight assignment.
However, the order in which the elements reveal themselves
is assumed to be random.
Soto~\cite{soto_2013_matroid} considered the variant where
not only the arrival order of the elements
is assumed to be uniformly random but also the assignment
of the weights to the elements, and presented
a $5.7187$-competitive algorithm for this case.
More precisely, in this model, the weights can still
be chosen by an adversary, but are then assigned uniformly
at random to the elements of the matroid.
Building upon earlier work of Soto~\cite{soto_2011_matroid},
Vondr\'ak and Oveis
Gharan~\cite{oveisgharan_2011_variants} showed that a
$40e/(e-1)$-competitive algorithm can even be obtained when the
arrival order of the elements is adversarial and the
assignment of weights remains 
uniformly at random. This was later improved to
a $16e/(e-1)$-competitive algorithm by Soto~\cite{soto_2013_matroid}.
Hence, this model is somehow the opposite of the classical
matroid secretary problem, where assignment is adversarial
and arrival order is random.

However, so far, no progress has been made in variants
with adversarial assignment.
One such variant, suggested shortly after
the introduction of the matroid secretary
problem~\cite{kleinberg_2012_personal},
assumes that the appearance order of elements 
can be chosen by the algorithm. More precisely,
in this model, which we call the \emph{free order model},
whenever a next element has to reveal itself,
the algorithm can choose the element to be
revealed.
For example, in the above network connection problem, one could
decide at each step which is the next client to reveal
its price, by using for this decision the network structure
and the elements observed so far.
A main further complication when dealing with
adversarial assignments---as in the free order
model---contrary to random assignment, is that the knowledge
of the initial structure of the matroid seems to
be of little help. This is due to the fact that
an adversary can assign a weight of zero to most
elements of the matroid, and only give a non-negative
weight to a selected subset $A\subseteq N$ of elements.
Hence, the problem essentially reduces to the restriction
$M\vert_A$ of the matroid $M$ over the elements $A$.
However, the structure of $M\vert_A$ is essentially
impossible to guess from $M$.
This is in stark contrast to models with
random assignment, e.g., in the model considered
by Soto, the mentioned $2e^2/(e-1)$-competitive
algorithm
exploits the given structure
of the matroid $M$, by partitioning $N$ and solving
a standard single secretary problem on each part of
the partition.
Different approaches are needed for
adversarial weight assignments.

In this paper we are interested in the following two questions.
First, is there an $O(1)$-competitive algorithm
for the free order model?
Second, can we get a better understanding
of the laminar case of the classical secretary problem,
with the goal to find considerably stronger and simpler
procedures?

As is common in this context, we use competitive analysis
to judge the quality of algorithms.
More precisely,  an algorithm is \emph{$c$-competitive}
if it returns a (random) solution whose expected value is at
least $\frac{1}{c} \OPT$, where $\OPT$ is the value of
an \emph{offline} optimum solution, i.e., a maximum weight
independent set.
Hence, the goal is to find $c$-competitive algorithms with
$c\geq 1$ being as close as possible to $1$.

\subsubsection*{Our results and techniques}
We present a $4$-competitive algorithm for the free
order model, thus obtaining the first $O(1)$-competitive
algorithm for a variant of the matroid secretary problem with
adversarial weight assignment, without any restriction on
the underlying matroid. 
This algorithm is in particular applicable 
to the previously mentioned network
connection problem, when the order, in which the network
operator negotiates with the clients, can be chosen.

On a high level, our algorithm follows a quite intuitive
idea, which,
interestingly, does not work in the traditional
matroid secretary problem.
In a first phase,
we draw each element with probability~$0.5$ to
obtain a set $A\subseteq N$, without
selecting any element of $A$.
Let $\OPT_A$ be the best offline solution in $A$.
We call an element $f\in N\setminus A$ \emph{good},
if it
can be used to improve $\OPT_A$, in the sense that
either $\OPT_A \cup \{f\}$ is independent
or there is an element $g\in \OPT_A$ such that
$(\OPT_A \setminus \{g\}) \cup \{f\}$ is independent
and has a higher value than $\OPT_A$.
In the second phase, we go through
the remaining elements $N\setminus A$,
drawing element by element in a well-chosen way
to be specified soon.
We accept an element $f\in N\setminus A$ if it is good and does not destroy
independence when added to the elements accepted so far.
Our approach fails if elements are drawn randomly in
the second phase. The main problem when drawing randomly,
is that we may accept good elements of relatively low value that
may later \emph{block} some high-valued good elements, in the
sense that they cannot be added anymore
without destroying independence of the selected
elements.
To overcome this problem, we determine after the first
phase a specific order of how elements will be drawn
in the second phase.
The idea is to first draw elements of $N\setminus A$ that are
in the span of elements of $A$ of high weight.
More precisely, let $A=\{a_1,\dots, a_m\}$ be the numbering of
the elements of $A$ according to decreasing weights.
In the second phase
we start by drawing elements of $(N\setminus A)\cap \Span(\{a_1\})$,
then $(N\setminus A)\cap \Span(\{a_1,a_2\})$, and so
on\footnote{We recall that $\Span(S)$ for $S\subseteq N$
is the unique maximal set $U\supseteq S$ with the same rank as $S$.}.
One particular situation, where the above ordering becomes 
very intuitive, is if there is a set $S\subseteq N$
with a high density of high-valued elements.
In this case it is likely that many elements of $S$ are part of
$A$. Hence, high-valued elements of $A$
span further high-valued elements in $S$.
Thus, by the above order, we are likely to
draw high-valued elements of $S$ early, before they can
be blocked by the inclusion of lower-valued elements.

Similar to previous secretary algorithms,
we show that
our algorithm is $O(1)$-competitive by proving that
each element $f\in \OPT$ of the global offline optimum $\OPT$
will be chosen with probability at least $1/4$.
However, the way we prove this is based on a novel approach.
Broadly speaking, we show that
an element $f \in \OPT$ gets selected if
additionally to $f\not\in A$, the following
property holds: either
$f\not\in \Span((N\setminus A)\setminus \{f\})$,
or the maximum value $\beta \geq 0$ such that
$f$ is spanned by elements in $(N\setminus A)\setminus\{f\}$
of weight $\geq \beta$ is smaller than the
maximum value $\alpha \geq 0$
such that $f$ is spanned by elements in $N\cap A$ of weight
$\geq \alpha$.
Exploiting that the
distributions of $A$ and $N\setminus A$ are identical,
we show that the above conditions happens with
probability at least $1/4$.

In an earlier short version of this
paper~\cite{jaillet_2013_advances}, we only proved that our
algorithm is $9$-competitive. Our proof was later
refined and simplified by Azar, Kleinberg and
Weinberg~\cite{azar_2014_prophet} to show $4$-competitiveness
of our procedure.
Due to this recent development, we present here the refined
analysis of~\cite{azar_2014_prophet}.
We are thankful to the authors of~\cite{azar_2014_prophet}
for their agreement to include this analysis in the present
paper.

Furthermore, we present a new approach to deal with laminar
matroids in the classical matroid secretary model.
Our technique leads to a $3\sqrt{3} e \approx 14.12$-competitive
procedure,
thus considerably improving on
the $16000/3\approx 5333$-competitive algorithm 
of Im and Wang~\cite{im_2011_secretary}.
Our main contribution here is to present a simple way to transform
the matroid secretary problem on a laminar matroid $M$ to
one
on a unitary partition matroid%
\footnote{A \emph{unitary partition matroid} is a partition matroid
where at most one element can be chosen in each set of the partition.}
$M_P$ by losing only a
small constant factor of $3\sqrt{3}\approx 5.2$.
The secretary problem on $M_P$ can then
simply be solved by applying the classical $e$-competitive algorithm
for the standard secretary problem to each partition of $M_P$.
We first observe a constant fraction of all elements, on the basis
of which a partition matroid $M_P$ on the remaining elements is
then constructed.
To assure feasibility, $M_P$ is defined such that
each independent set of $M_P$ is also an independent set of $M$.
To best convey the main ideas of our procedure, 
we first present
a very simple method to obtain a
weaker $27e/2\approx 36.7$-competitive algorithm,
which already improves considerably on
the $16000/3$-competitive algorithm of Im and Wang.
The $3\sqrt{3} e$-competitive algorithm is then
obtained through a strengthening
of this approach by using a stronger partition matroid $M_P$ and
a tighter analysis.

A further advantage of our procedure for laminar matroids is the
fact that it leads to $O(1)$-competitive algorithms in the context
of single-sample matroid prophet inequalities, which in turn implies
strong algorithms for order-oblivious posted
pricing mechanisms, as shown by Azar,
Kleinberg and Weinberg~\cite{azar_2014_prophet}.
More precisely, prophet inequalities are a setting that is closely related
to the matroid secretary problem. The key difference is that the weight
of each element comes from a distribution that depends on the element,
and depending on the setting may or may not be known in advance.
In single-sample prophet inequalities, one only knows a single sample
from each distribution, and the order in which the elements arrive is
adversarial, which is another key difference to the classical matroid
secretary problem.
It was shown in~\cite{azar_2014_prophet} that an $\alpha$-competitive
matroid secretary algorithm can be transformed into an
$\alpha$-competitive algorithm for single-sample prophet inequalities,
if the secretary algorithm is \emph{order-oblivious}.
Loosely speaking, an order-oblivious procedure is
one that consists of two phases, where in a first phase a subset
of the elements is observed without choosing any element, and
furthermore, the competitive ratio does not dependent on the
order in which elements appear in the second phase. Hence, the
algorithm does not need the random order assumption during the second
phase.
Contrary to the previous $O(1)$-competitive laminar secretary
algorithm~\cite{im_2011_secretary}, and also a subsequently
introduced $O(1)$-competitive algorithm for this
case~\cite{ma_2013_simulatedSTACS}, our algorithm is order-oblivious.
We refer the reader
to~\cite{azar_2014_prophet} for more information on order-oblivious
algorithms and single-sample prophet inequalities.
Furthermore,~\cite{azar_2014_prophet} also discusses the implications
of our algorithm for laminar matroids in this context.

We remark that the algorithms we present do
not need to observe the exact
weights of the items when they reveal themselves, but only need
to be able to \emph{compare} the weights of elements observed so far.
This is a common feature of many matroid secretary algorithms
and matroid algorithms more generally.

To simplify the exposition, we assume that all weights
are distinct, i.e., they induce a linear order on the elements.
This implies in particular, that there is a unique
maximum weight independent set.
The general case with possibly equal weights
easily reduces to this case
by breaking ties arbitrarily between elements of equal weight,
to obtain a linear order.

\subsubsection*{Related work}
Recently,
matroid secretary problems
with submodular objective functions have been considered.
For this setting, $O(1)$-competitive procedures have been found
for knapsack constraints, uniform matroids, and,
if the submodular objective is furthermore monotone,
for partition matroids, and more generally for
intersections of laminar matroids,
and transversal matroids (see~\cite{bateni_2010_submodular,
feldman_2011_improved, gupta_2010_constrained,
ma_2013_simulatedSTACS}).

Additionally, variations of the matroid secretary problem have
been studied with restricted knowledge on the underlying
matroid type. This includes the case where
no prior knowledge of the underlying
matroid is assumed except for the size of the ground set.
Or even more extremely, the case
without even knowing the size of the ground set.
For more information on such variations
we refer to the excellent overview
in~\cite{oveisgharan_2011_variants}.

\subsubsection*{Subsequent results}
We would like to highlight that very recently,
after a previous version~\cite{jaillet_2012_advances}
of this article,
Ma, Tang and Wang~\cite{ma_2012_simulated, ma_2013_simulatedSTACS} further
improved the competitive ratio for
the secretary problem on laminar matroids by presenting
a $9.6$-competitive algorithm.
They use an interesting and natural algorithmic idea, including
elements only if they are part of the offline optimum
of all elements seen so far.
The description of their algorithm is nice and elegant,
however, its analysis is somewhat involved.
Unfortunately, their algorithm
is not order-oblivious and therefore cannot be used in the context
of single-sample prophet inequalities.

\subsubsection*{Organization of the paper}
Our $4$-competitive algorithm for the free order
model is presented in Section~\ref{sec:freeOrder}.
Section~\ref{sec:laminar} discusses our algorithms for
the classical matroid secretary problem.
We start by presenting in Section~\ref{subsec:simpleLam} 
our simple $27e/2$-competitive method, and then
show in Section~\ref{subsec:betterLam} how to strengthen
the algorithm and its analysis
to obtain the claimed $3\sqrt{3}e$-competitiveness.

\section{A $4$-competitive algorithm for the
free order model}\label{sec:freeOrder}

To simplify the writing we use ``$+$'' and ``$-$''
for the addition and subtraction of single elements
from a set, i.e., $S+f-g = (S\cup\{f\})\setminus \{g\}$.
Furthermore, for $k\in \mathbb{Z}_{\geq 1}$ we use
the shorthand $[k]:=\{1,\dots, k\}$.
Algorithm~\ref{alg:fo} describes
our $4$-competitive algorithm
for the free order model.

\begin{algorithm}
\caption{A $4$-competitive algorithm
for the free order model.}
\label{alg:fo}
\begin{enumerate}
\item\label{algitem:foSamplingStep} \textbf{Draw}
each element with probability~$0.5$ to obtain
$A\subseteq N$, without selecting any element of $A$.
We number the elements of $A=\{a_1,\dots, a_m\}$ in
decreasing order of weights.
Define $A_i=\{a_1,\dots, a_i\}$, with $A_0=\emptyset$.\\
\textbf{Initialize:} $I\leftarrow \emptyset$.

\item\label{algitem:foSelectionStep}
\textbf{For} $i=1$ to $m$:\\
  \hspace*{1em} \textbf{draw one by one} (in any order) all
  elements $f\in (\Span(A_{i})\setminus \Span(A_{i-1}))\setminus A$:\\
    \hspace*{2em} \textbf{if} $I+f\in \mathcal{I}$ and
    $\weight(f)>\weight(a_i)$, \textbf{then} $I=I+f$.\\
\textbf{For} all remaining elements $f\in N\setminus \Span(A)$
(drawn in any order):\\
  \hspace*{1em} \textbf{if} $I+f\in \mathcal{I}$,
    \textbf{then} $I=I+f$.\\
  \textbf{Return} $I$
\end{enumerate}
\end{algorithm}

To analyze Algorithm~\ref{alg:fo}, we introduce some additional
notation. Let $\{e_1,\dots, e_n\}=N$ be the numbering of the
elements of the ground set satisfying
$\weight(e_1) > \dots > \weight(e_n)$. Furthermore, for
each $j\in [n]$, we define $N_j := \{e_1,\dots, e_j\}$.

As mentioned previously, a \emph{good} element
$f\in N\setminus A$ is an element that allows for improving
the maximum weight independent set in $A$,
i.e., the unique maximum weight independent set in
$A+f$ includes $f$.
An element $f$ being good thus means that it gets selected
when applying the greedy algorithm to $A+f$. Hence, $f$ is good
if and only if $f\not\in \Span(\{a\in A\mid \weight(a) > \weight(f)\})$,
which can be rephrased as $f$ is good if either $f\not\in \Span(A)$,
or if there is an index $i\in [m]$ such that
$f\in \Span(A_i)\setminus \Span(A_{i-1})$
and $\weight(f)>\weight(a_i)$.
Hence, our algorithm indeed only accepts good elements.
Furthermore, whenever any element $f$ of the offline
optimum $\OPT$ is considered in some
iteration $i\in [m]$ in the
first for-loop of step~\ref{algitem:foSelectionStep}---i.e.,
$f\in (\Span(A_i)\setminus \Span(A_{i-1}))\setminus A$---then
we always have $\weight(f) > \weight(a_i)$.
Hence, an element $f\in \OPT$ gets selected 
by Algorithm~\ref{alg:fo}
if and only if $I+f\in \mathcal{I}$, where $I$ is the set
of already selected elements at the time when $f$
is considered.

To show that Algorithm~\ref{alg:fo} is
$4$-competitive,
we show that each element $f\in \OPT$ 
will be contained in the
set $I$ returned by the algorithm
with probability at least $1/4$.
Hence, let $f\in \OPT$, and define
\begin{align*}
j_1 &:= \argmin\{j\in [n] \mid f\in \Span((N_j \cap A)-f)\}, \\
j_2 &:= \argmin\{j\in [n] \mid f\in \Span((N_j \setminus A)-f)\}, 
\end{align*}
where, if $f\not\in \Span(A-f)$ we set $j_1=\infty$, and
and $f\not\in\Span((N\setminus A)-f)$ then
$j_2=\infty$.
Notice that both $j_1$ and $j_2$ are random variables that depend
on the random set $A$. The following lemma provides a simple
property under which $f$ gets selected.
\begin{lemma}\label{lem:okIfJ1LeqJ2}
If $j_1 \leq j_2$ and $f\not\in A$
then $f$ gets selected by Algorithm~\ref{alg:fo}.
\end{lemma}
\begin{proof}
We first handle the case $j_2=\infty$.
Consider the moment when $f$ is considered in the
second step of Algorithm~\ref{algitem:foSelectionStep},
either in the first or second for-loop, and let $I$
be the elements selected so far by the algorithm.
As discussed above, since $f\in\OPT$, we only have
to show $I+f\in \mathcal{I}$ for $f$ to be selected,
which holds because
\begin{equation*}
f \overset{j_2=\infty}{\not\in}
  \Span((N\setminus A)-f)
  \overset{I\subseteq (N\setminus A)-f}{\supseteq}
  \Span(I),
\end{equation*}
and hence, $f\not\in \Span(I)$.

Now assume $j_2 < \infty$, and therefore also
$j_1 < \infty$ since $j_1\leq j_2$.
Consider the moment when $f$ is considered in the
second step of Algorithm~\ref{alg:fo}, and let $I$
be the set of elements selected so far by the algorithm.
Notice that since $j_1 < \infty$, we have $f\in \Span(A-f)$,
and therefore, $f$ is considered at some iteration $i\in [m]$
during the first for-loop
of step~\ref{algitem:foSelectionStep} of the algorithm.
Furthermore, $i$ is the smallest index in $[m]$ such
that $f\in \Span(A_i)$, and thus, $a_i = e_{j_1}$.
Hence, $I$ only contains elements of weight
strictly larger than $\weight(e_{j_1})$, i.e.,
$I\subseteq N_{j_1-1}$. Furthermore,
since $I\subseteq N\setminus A$, this implies
$I\subseteq N_{j_1-1}\setminus A$.
Again we have $I+f\in \mathcal{I}$ since 
\begin{equation*}
f \overset{\text{def.~of $j_2$}}{\not\in}
 \Span((N_{j_2-1}\setminus A)-f)
 \overset{j_1\leq j_2}{\supseteq} \Span((N_{j_1-1} \setminus A)-f)
 \overset{I\subseteq (N_{j_1-1}\setminus A)-f}{\supseteq} \Span(I).
\end{equation*}

\end{proof}

Leveraging Lemma~\ref{lem:okIfJ1LeqJ2}, we can now prove
the correctness of the algorithm by showing that
$j_1\leq j_2$ with probability at least $0.5$.

\begin{theorem}
Algorithm~\ref{alg:fo} selects each element 
$f\in \OPT$ with probability at least $1/4$,
and is therefore $4$-competitive.
\end{theorem}
\begin{proof}
The key observation is that for every $S\subseteq N$,
either $j_1\leq j_2$ for $A=S$ or $j_1\leq j_2$ for
$A=N\setminus S$. Since the two events
$A=S$ and $A=N\setminus S$
occur with the same probability, we obtain
$\Pr[j_1\leq j_2] \geq 0.5$.
Furthermore, whether $f\in A$ or not is independent
of the two random variables $j_1$ and $j_2$.
Thus, 
\begin{equation*}
\Pr[f\not\in A \text{ and } j_1\leq j_2]
  = \Pr[f\not\in A] \cdot \Pr[j_1\leq j_2]
  \geq \frac{1}{2}\cdot \frac{1}{2} = \frac{1}{4},
\end{equation*}
and by Lemma~\ref{lem:okIfJ1LeqJ2}, the probability
of $f$ being selected by Algorithm~\ref{alg:fo} is
at least $1/4$.
\end{proof}

\section{Classical
secretary problem for
laminar matroids}\label{sec:laminar}

Let $M=(N,\mathcal{I})$ be a laminar matroid whose
constraints are defined by the laminar family
$\mathcal{L}\subseteq 2^N$ with upper bounds $b_L$
for $L\in\mathcal{L}$
on the number of elements that can be chosen from
$\mathcal{L}$, i.e.,
$\mathcal{I}=\{I\subseteq N \mid |I\cap L|\leq b_L
\;\forall L\in \mathcal{L}\}$.
Without loss of generality
we assume $b_L\geq 1$ for $L\in\mathcal{L}$, since otherwise
we can simply remove all elements of $L$ from $M$.
Furthermore, we assume $N\in \mathcal{L}$,
since otherwise a redundant constraint $|I\cap N|\leq b_N$
can be added
by choosing a sufficiently large right-hand side $b_N$.

\subsection{A simple $27e/2$-competitive
algorithm for the laminar secretary problem}\label{subsec:simpleLam}

To reduce the matroid secretary problem on $M$ to a
problem on a partition matroid, we first number the
elements $N=\{f_1,\dots, f_n\}$ 
such that for any set $L\in \mathcal{L}$, the elements
in $L$ are numbered consecutively, i.e., $L=\{f_p,\dots, f_q\}$
for some $1\leq p < q \leq n$.
Figure~\ref{fig:lamNumbering} shows an example of such a
numbering.

\begin{figure}[h!t]
$$\includegraphics[scale=0.9]{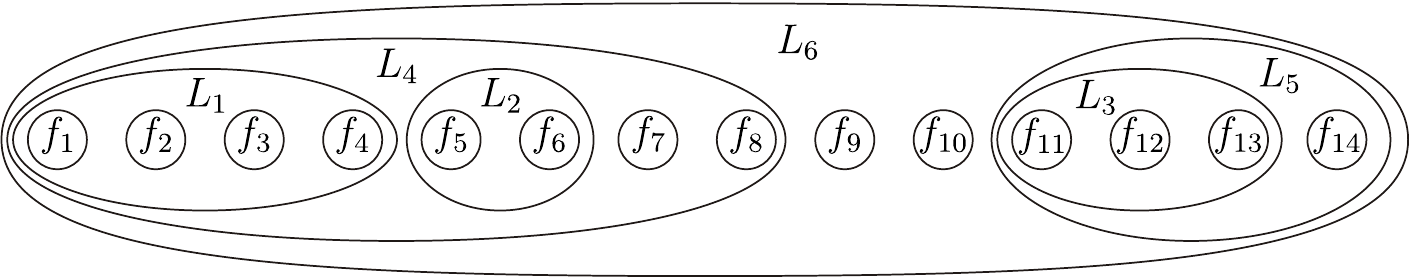}$$
\vspace{-1em}
\caption{An example of a numbering of the elements of the
ground set such that each set
$L\in \mathcal{L}=\{L_1,\dots, L_6\}$
of the laminar family contains consecutively
numbered elements.
}
\label{fig:lamNumbering}
\end{figure}

For the sake of exposition, we
start by presenting a conceptually simple
algorithm and analysis, based on the
introduced numbering of the ground set, that leads
to a competitive ratio of $27e/2$.
The claimed $3\sqrt{3}e$-competitive algorithm follows the
same ideas, but strengthens both the approach and analysis.
Algorithm~\ref{alg:simpleLaminar} describes our
$27e/2$-competitive procedure.
%
\begin{algorithm}
\caption{A $27e/2$-competitive algorithm for laminar matroids.}
\label{alg:simpleLaminar}
\begin{enumerate}
\item \textbf{Observe} $\Binom(n,2/3)$
elements of $N$, which we denote by $A\subseteq N$.\\
\textbf{Determine} maximum weight independent set
$\OPT_A=\{f_{i_1},\dots, f_{i_p}\}$ in $A$
where $1 \leq i_1 < \dots < i_p \leq n$.
Define $P_j = \{f_k \mid k\in \{i_{j-1},\dots, i_{j}\}\}\setminus A$
for $j\in \{1,\dots, p+1\}$, where we set $i_0=0, i_{p+1}=n$.
Let
\begin{align*}
\mathcal{P}_{\odd}(A) &= \{P_j\mid
  j\in [p+1], j \text{ odd}\},\\
\mathcal{P}_{\even}(A) &= \{P_j\mid
  j\in [p+1], j \text{ even}\}.
\end{align*}
\textbf{If} $\OPT_A=\emptyset$ \textbf{then} set $\mathcal{P}=\{N\setminus A\}$,\\
\textbf{else set} $\mathcal{P}=\mathcal{P}_{\odd}(A)$
with probability $0.5$,
otherwise set $\mathcal{P}=\mathcal{P}_{\even}(A)$.

\item\label{item:sLPhase2} \textbf{Apply} to each set $P\in \mathcal{P}$
an $e$-competitive classical secretary algorithm
to obtain an element $g_P\in P$.\\
\textbf{Return} $\{g_P\mid P\in \mathcal{P}\}$.
\end{enumerate}

\vspace*{-1em}
\end{algorithm}
Notice that applying a standard secretary
algorithm to the sets of $\mathcal{P}$
in step~\ref{item:sLPhase2} can easily be performed
by running $|\mathcal{P}|$ many $e$-competitive secretary algorithms
in parallel, one for each set $P\in \mathcal{P}$.
Elements are drawn one by one in the second phase, and they are
forwarded to the secretary algorithm corresponding to the set $P$
that contains the drawn element, and are discarded if no set
of $\mathcal{P}$ contains the element.
Furthermore, observe that $A$ contains each element of $N$
independently with probability $2/3$.

We start by observing that Algorithm~\ref{alg:simpleLaminar}
returns an independent set.

\begin{lemma}
Let $A\subseteq N$ with $\OPT_A\neq \emptyset$ and let
$\mathcal{P}\in \{\mathcal{P}_{\even}(A), \mathcal{P}_{\odd}(A)\}$.
For each $P\in \mathcal{P}$, let $g_p$ be any element in $P$.
Then $\{g_P\mid P\in \mathcal{P}\}\in \mathcal{I}$.
\end{lemma}
\begin{proof}
Let $I=\{g_P\mid P\in \mathcal{P}\}$ be a set
as stated in the lemma.
Notice that for any two elements $f_k, f_\ell \in I$ with
$k< \ell$ we have
$|\OPT_A \cap \{f_k,f_{k+1}, \dots, f_{\ell}\}|\geq 2$.
Now consider a set $L\in \mathcal{L}$ corresponding to one
of the constraints of the underlying laminar matroid.
By the above observation and since $L$ is consecutively numbered,
at least one of the following holds:
(i) $|L\cap I|=1$, or (ii) $|L\cap \OPT_A |\geq |L\cap I|$.
If case (i) holds, then the constraint corresponding to $L$
is not violated since we assumed $b_L\geq 1$.
If (ii) holds, then $L$ is also not violated since
$|L\cap I|\leq |L\cap \OPT_A| \leq b_T$
because $\OPT_A \in \mathcal{I}$.
Hence $I\in \mathcal{I}$.
\end{proof}

\begin{theorem}\label{thm:analysisSimpleLaminar}
Algorithm~\ref{alg:simpleLaminar} is
$27e/2$-competitive for the laminar matroid secretary
problem.
\end{theorem}
\begin{proof}
Let $\OPT\in \mathcal{I}$ be the maximum weight independent
set in $N$, i.e., the offline optimum. Furthermore, let $I$
be the set returned by Algorithm~\ref{alg:simpleLaminar},
and let $f\in \OPT$.
We say that $f$ is \emph{solitary}
if $\exists\: P\in \mathcal{P}$ with
$P\cap \OPT=\{f\}$.
Similarly we call $P\in \mathcal{P}$ \emph{solitary}
if $|P \cap \OPT|=1$.
We prove the theorem by showing that each element $f\in\OPT$
is solitary with probability $\geq 2/27$.
This indeed implies the theorem since we can do the following type
of accounting. Let $X_{f}$ be the random variable which
is zero if $f$ is not solitary, and
otherwise 
equals the weight of the element $g\in I$ that was chosen
by the algorithm from the set $P$ that contains $f$.
By only considering the weights of elements chosen
in solitary sets $\mathcal{P}$ we obtain
\begin{equation}\label{eq:expGood}
\E[w(I)]\geq \sum_{f\in \OPT} \E[X_{f}].
\end{equation}
However, if each element $f\in \OPT$ is solitary with
probability $2/27$, we obtain
$\E[X_{f}] \geq \frac{2 w(f)}{27 e}$, because the classical
secretary algorithm will choose with probability $1/e$
the maximum weight element of the set $P$ that contains the
solitary element $f$.
Combining this with~\eqref{eq:expGood}
yields $\E[w(I)]\geq \frac{2}{27 e} w(\OPT)$ as desired.
Let us then show that each $f \in \OPT$ is solitary
with probability $\geq 2/27$.

Let $f_i\in \OPT$. We assume that $\OPT$ contains an
element with a lower index than $i$ and one with a higher index
than $i$.
The cases of $f_i$ being the element with highest or lowest
index in $\OPT$ follow analogously.
Let $f_j\in \OPT$ by the element
of $\OPT$ with the largest index $j<i$. Similarly,
let $f_k\in \OPT$ be the element of $\OPT$
with the smallest index $k>i$.
One well-known matroidal property that we use is
$\OPT\cap A\subseteq \OPT_A$.
Hence, if $f_j, f_k \in A$ then $f_j,f_k\in \OPT_A$,
and if furthermore $f_i\not\in A$, then $f_i$ will be the
only element of $\OPT$ in the set
$P\in \mathcal{P}_{\odd}(A)\cup \mathcal{P}_{\even}(A)$
that contains $f_i$. Hence, if the coin flip in
Algorithm~\ref{alg:simpleLaminar} chooses the family
$\mathcal{P}\in \{\mathcal{P}_{\odd}(A), \mathcal{P}_{\even}(A)\}$
that contains $P$, then $f_i$ is solitary.
To summarize, $f_i$ is solitary if $f_j,f_k \in A$,
$f_i\not\in A$ and the coin flip for $\mathcal{P}$
turns out right. This happens with probability
$\left(\frac{2}{3}\right)^2 \cdot%
\left(1-\frac{2}{3}\right) \cdot \frac{1}{2} = \frac{2}{27}$.
\end{proof}

\subsection{A $3\sqrt{3}e$-competitive algorithm for the
laminar matroid secretary problem}\label{subsec:betterLam}

One conservative aspect of the proof of
Theorem~\ref{thm:analysisSimpleLaminar} is that we only
consider the contribution of solitary elements.
Additionally, a drawback of Algorithm~\ref{alg:simpleLaminar}
itself is that about half of the elements of $N\setminus A$
are ignored as we only select from either
$\mathcal{P}_{\odd}(A)$ or $\mathcal{P}_{\even}(A)$.
In this section, we address these two
weaknesses to obtain a $3\sqrt{3}e$-competitive
algorithm.

We start by describing a stronger way to define a partition $\mathcal{P}$
of $N\setminus A$ and reduce the problem to a matroid secretary
problem on the unitary partition matroid defined on $\mathcal{P}$.

For any independent set $I\in \mathcal{I}$, we define a partition
$\widetilde{\mathcal{P}}(I)$ of $N$ as follows.
If $I=\emptyset$, we set $\widetilde{\mathcal{P}}(I)=\{N\}$.
Otherwise $\widetilde{\mathcal{P}}(I)$ contains a
set $N_f\subseteq N$ for each element $f\in I$,
i.e., $\widetilde{\mathcal{P}}(I)=\{N_f \mid f\in I\}$.
To define the partition $\widetilde{\mathcal{P}}(I)$,
we specify to which
set $N_{f}$ an element $f_i\in N$ belongs.
Let $L\in \mathcal{L}$ be the smallest set that
contains $f_i$ and such that $L\cap I \neq \emptyset$.
Such a set must exist since $N\in \mathcal{L}$ by assumption.
If $L\cap I$ contains at least one element $f_j$
with $j\leq i$, then let $j$
be the largest index such
that $j\leq i$ and $f_j\in L\cap I$.
Otherwise let $j$ be the smallest index satisfying $j> i$
and $f_j\in L\cap I$.
We assign the element $f_i$ to $N_{f_j}$.

Notice that in any case, $j$ is either the largest
index $j\leq i$ with $f_j\in I$ or the smallest index $j>i$
with $f_j\in I$.
Again, we are interested to define a partition only on
elements $N\setminus A$ not drawn in the first phase.
We therefore define for any $A\subseteq N$ the partition
$\mathcal{P}(A)=\{\widetilde{P}\setminus A
\mid \widetilde{P}\in \widetilde{\mathcal{P}}(\OPT_A)\}$.
Algorithm~\ref{alg:laminarImproved} describes our
$3\sqrt{3}e$-competitive procedure.

\begin{algorithm}
\caption{A $3\sqrt{3} e$-competitive algorithm
for laminar matroids.}
\label{alg:laminarImproved}
\begin{enumerate}
\item \textbf{Observe} $\Binom(n,1/\sqrt{3})$
elements of $N$, which we denote by $A\subseteq N$.\\
\textbf{Determine} maximum weight independent set
$\OPT_A$ in $A$.
\item \textbf{Apply} to each set $P\in \mathcal{P}(A)$
an $e$-competitive
classical secretary algorithm to obtain
$g_P\in P$.\\
\textbf{Return} $\{g_P\mid P\in \mathcal{P}(A)\}$.
\end{enumerate}
\end{algorithm}

We first show that the set returned by
Algorithm~\ref{alg:laminarImproved} is indeed independent.
For this, we start by observing a basic property of the
sets $N_f$ forming the underlying partition
$\widetilde{\mathcal{P}}(\OPT_A)=\{N_f\mid f\in \OPT_A\}$.
\begin{lemma}\label{lem:consecInd}
Let $I\in \mathcal{I}$ with $I\neq\emptyset$.
Each set $N_{f_i}$ of the partition
$\widetilde{\mathcal{P}}(I)=\{N_{f_i}\mid f_i\in I\}$
is of the form $N_{f_i} = \{f_j, f_{j+1}, \dots, f_k\}$
for some $1\leq j \leq i \leq k \leq n$.
\end{lemma}
\begin{proof}
By definition of $N_{f_i}$, we clearly
have $f_i\in N_{f_i}$. Hence, all that remains to be shown
is that whenever $f_p \in N_{f_i}$, then $f_q\in N_{f_i}$
for any $q$ between $i$ and $p$, i.e., either $i< q < p$ or
$p< q < i$. In the following we distinguish these
two cases. For any element $f\in N$, we denote by
$L_f\in \mathcal{L}$ the smallest set $L\in \mathcal{L}$
that contains $f$ and satisfies $I\cap L\neq \emptyset$.

\textbf{Case $p<q<i$.}
Since $f_p\in N_{f_i}$,
there is no element $f_\ell \in L_{f_p}\cap I$ with $\ell< i$.
Furthermore, $f_p, f_i \in L_{f_p}$ implies $f_q\in L_{f_p}$,
because $L_{f_p}$ contains a sequence of
consecutively numbered elements.
As a consequence, there is also no element
$f_\ell\in L_{f_q}\cap I$ with $\ell<i$,
because $L_{f_q}\subseteq L_{f_p}$
due to laminarity and the fact that $L_{f_q}$ is
the smallest set in $\mathcal{L}$ containing $f_q$
and satisfying $I\cap L_{f_q}\neq \emptyset$.
Hence $f_q\in N_{f_i}$.

\textbf{Case $i<q<p$.}
As in the previous case we have
$f_i \in L_{f_q} \subseteq L_{f_p}$,
and there is no $\ell$ with $i<\ell<p$ such that $f_\ell\in I$,
using again $f_p\in N_{f_i}$.
Thus, $f_q\in N_{f_i}$.
\end{proof}

The next lemma implies that
Algorithm~\ref{alg:laminarImproved} returns an
independent set.

\begin{lemma}
Let $I\subseteq \mathcal{I}$ and let
$J\subseteq N$ with $|J\cap \widetilde{P}|\leq 1$
$\forall \widetilde{P}\in \widetilde{\mathcal{P}}(I)$.
Then $J\in \mathcal{I}$.
\end{lemma}
\begin{proof}
To show $J\in \mathcal{I}$ we fix
any $L\in\mathcal{L}$ and show that $J$ satisfies the constraint
imposed on $L$, i.e., $|J\cap L|\leq b_L$.
If $I\cap L=\emptyset$, then all elements in $L$ belong
to the same set of the partition $\widetilde{\mathcal{P}}(I)$.
Hence $|J\cap L|\leq 1$, and the constraint
corresponding to $L$ is not
violated since by assumption $b_L\geq 1$.
Hence, assume $I\cap L \neq \emptyset$.
Notice that in this case
every element in $L$ will be assigned to a set
$N_{f}$ for $f\in I\cap L$, i.e.,
\begin{equation}\label{eq:partL}
L \subseteq \bigcup_{f\in I\cap L} N_f.
\end{equation}
Since at most one element is chosen out
of each $N_f$ we have
\begin{equation*}
|J\cap L| \leq |I\cap L| \leq b_L,
\end{equation*}
where the second inequality follows from $I \in \mathcal{I}$.
\end{proof}

As the family $\mathcal{P}(A)$ consists of subsets of the partition
$\widetilde{\mathcal{P}}(\OPT_A)$, the above lemma implies:

\begin{corollary}
Algorithm~\ref{alg:laminarImproved} returns an independent set.
\end{corollary}

It remains to show the claimed competitiveness.
\begin{theorem}\label{thm:laminarImproved}
Algorithm~\ref{alg:laminarImproved} is
$3\sqrt{3} e$-competitive for the laminar
matroid secretary problem.
\end{theorem}

\begin{proof}

Let $\OPT_{\mathcal{P}(A)}$ be the optimum solution of the
matroid secretary problem on $N\setminus A$ constrained by the
partition matroid $\mathcal{P}(A)$. Let $I$ be the solution
returned by Algorithm~\ref{alg:laminarImproved}.
Since Algorithm~\ref{alg:laminarImproved} applies an
$e$-competitive secretary algorithm to each set of
$\mathcal{P}(A)$, we have
\begin{equation}\label{eq:compToPartMat}
\E[\weight(I)]\geq \frac{1}{e}\cdot\E[\weight(\OPT_{\mathcal{P}(A)})].
\end{equation}
For $f\in N\setminus A$, we denote by $P_f$ the set
in the family $\mathcal{P}(A)$ that contains $f$.
We have,

\begin{align}
\E[\weight(\OPT_{\mathcal{P}(A)})] &=
    \E\left[\sum_{P \in \mathcal{P}(A)}
    \max_{f \in P} \weight(f)\right]
 \geq \E\left[\sum_{\substack{P \in \mathcal{P}(A),\\
        |P\cap \OPT|\geq 1}}
      \max_{f \in P} \weight(f)\right]\notag\\
&\geq \E\left[ \sum_{\substack{P \in \mathcal{P}(A),\\
     |P\cap \OPT|\geq 1}} \;
  \sum_{f \in P\cap \OPT} \frac{\weight(f)}%
                               {|P\cap \OPT|}\right]\notag\\
&= \E\left[ \sum_{f\in \OPT\setminus A}
     \frac{\weight(f)}{|P_f \cap \OPT|}
   \right].\label{eq:accToOPT}
\end{align}

Similar to the proof of Theorem~\ref{thm:analysisSimpleLaminar}
we use an accounting based on the elements of the offline
optimum $\OPT$.
For each $f\in \OPT$ we define a random variable $Z(f)$
as follows:
\begin{equation*}
Z(f) =
\begin{cases}
0 &\text{if } f\in A,\\
\frac{1}{|P_f\cap \OPT|} &\text{otherwise}.
\end{cases}
\end{equation*}
Together with~\eqref{eq:compToPartMat} and~\eqref{eq:accToOPT}
we thus obtain
\begin{equation*}
\E[\weight(I)] \geq \frac{1}{e}
  \sum_{f\in \OPT} \weight(f) \E[Z(f)].
\end{equation*}
Hence, to show that Algorithm~\ref{alg:laminarImproved} is
$3\sqrt{3} e$-competitive, is suffices to show
\begin{equation}\label{eq:toProveZ}
\E[Z(f)]\geq \frac{1}{3\sqrt{3}} \quad \forall f\in \OPT.
\end{equation}

For proving~\eqref{eq:toProveZ}, we want
to be able to treat all elements $f_i\in \OPT$ the same way,
independently of the index $i$. In particular, we want to avoid
special treatments for indices $i$ that are close to the
border, i.e., either close to $1$ or $n$.
Therefore we make the following assumptions,
which do not change the way in which the algorithm behaves: assume
that there are infinitely many dummy coloop\footnote{A coloop is
an element that is in every base of the matroid, or in other
words, a coloop element can be added to any independent set
without destroying independence.} elements (with zero weight)
denoted as
$C=\{\dots, f_{-2}, f_{-1}, f_0\} \cup \{f_{n+1}, f_{n+2}, \dots\}$.
The new (infinite) laminar matroid $M'$ is associated to the laminar family
$\mathcal{L}' = \mathcal{L} \cup \{N\cup C\}$, where $N\cup C$ has
no bound on the cardinality.

The optimum $\OPT'$ of $M'$ equals $C$ union
the optimum $\OPT=\{f_{i_1}, \dots, f_{i_p}\}$ of the original matroid.
If we run the algorithm on this modified infinite
matroid---assuming that every element, original or dummy,
belongs to $A$ with probability $1/\sqrt{3}$---and
then  remove the dummy elements from its
output, we recover the output that we would have obtained had
we used the real matroid.

We fix an element $f_{i_r}\in \OPT$ and prove~\eqref{eq:toProveZ}
for this element in the following.
To have $f_{i_j}$ defined for every integer $j$, even outside
of $\{1,\dots p\}$, we set
$i_j=j$ for $j\leq 0$, and $i_j=n-p+j$ for $j>p$.
Hence, $\OPT'=\{f_{i_j}\mid j \text{ integer}\}$.
Furthermore, to simplify the exposition and to explain
later why $1/\sqrt{3}$ was chosen to be the probability of
including elements in $A$, we denote by $q$ the probability
that an element is contained in $A$.

For every pair of natural numbers $s,t\geq 0$,
define $\calE_{s,t}$ as the event that the following occurs
simultaneously:
\begin{compactenum}[(i)]
\item $f_{i_r} \notin A$,
\item $f_{i_{r-1-s}}$ is the last element
of $\OPT'$ before $f_{i_r}$ that is in $A$, and
\item $f_{i_{r+1+t}}$ is the first element in $\OPT'$
after $f_{i_r}$ that is in $A$.
\end{compactenum}
In other word, $\calE_{s,t}$ is the event that
$f_{i_{r-(s+1)}}\in A$; $f_{i_{r-s}}, \dots,
f_{i_{r+t}} \not\in A$; and $f_{i_{r+(t+1)}} \in A$.

From this point on we condition on the event $\calE_{s,t}$.
Consider $\mathcal{L}'_{f_{i_r}}=%
\{L\in \mathcal{L}'\mid f_{i_r}\in L\}$.
Since $\mathcal{L}'$ is a laminar family,
$\mathcal{L}'_{f_{i_r}}$ is a chain.
Let $L\in \mathcal{L}'_{f_i}$ be the smallest
set in $\mathcal{L}'_{f_i}$ with
$(L\cap \OPT')\setminus A\neq \emptyset$; or equivalently,
$\{f_{i_{r-(s+1)}}, f_{i_{r+(t+1)}}\}\cap L\neq\emptyset$.
We claim that
\begin{align}\label{eq:bigequation}
  \E[Z(f_{i_r})\, |\, \calE_{s,t}] \geq
     q \sum_{k=0}^\infty \frac{1}{s+t+1+k}(1-q)^{k}.
\end{align}
To prove~\eqref{eq:bigequation}, we distinguish two cases:
(a) $f_{i_{r-(s+1)}}\not\in L$ and (b) $f_{i_{r-(s+1)}}\in L$.

In the first case, let $K\geq 0$ be the random variable counting
the number of consecutive elements in $(f_{i_j})_j$ immediately
after $f_{i_{r+t+1}}$ that are not contained in $A$. In other words,
$f_{i_{r+(t+1)+K+1}}$ is the first element of $\OPT'$ after
$f_{i_{r+(t+1)}}$ that is in $A$.
Note that conditioned on $\calE_{s,t}$ and on the
variable $K$, the set $P\in \mathcal{P}(A)$ to which
$f_{i_r}$ belongs must be a subset of
$Q=\{f_{i_{r-(s+1)}+1}, \dots, f_{i_{r+t+K+2}-1}\}$.

In particular,
$Q\cap \OPT' \subseteq \{f_{i_{r-s}},\dots, f_{i_{r+t+K+1}}\}$.
Recalling that $f_{i_{r+t+1}} \in A$ and
$P(f_{i_r})\subseteq N\setminus A$, we conclude
$|P(f_{i_r})\cap \OPT'| \leq |Q\cap \OPT'|-1$, and hence
\begin{equation*}
Z(f_{i_r}) = \frac{1}{|P(f_{i_r})\cap \OPT'|}
  \geq \frac{1}{|Q\cap \OPT| - 1} = \frac{1}{t+s+1+K}\;.
\end{equation*}
Therefore,
\begin{align*}
\E[Z(f_{i_r})\, |\, \calE_{s,t}]
  &\geq \sum_{k=0}^\infty \E[Z(f_{i_r})\,
    |\, \calE_{s,t}, K=k]\cdot\Pr(K=k)\\
  &\geq \sum_{k=0}^\infty \frac{1}{t+s+1+k} q(1-q)^k,
\end{align*}
which proves~\eqref{eq:bigequation} for case~(a).
The proof of the claim in case~(b) is analogous, but in
that case we define $K\geq 0$ as the random variable counting
the number of consecutive elements in $(f_{i_j})_j$
immediately before $f_{i_{r-(s+1)}}$ that are outside $A$.

Based on~\eqref{eq:bigequation}, we can conclude the proof
of the theorem as follows. Since all events
$(\calE_{s,t})_{s,t\geq 0}$ are disjoint
and $\Pr(\calE_{s,t}) = q^2(1-q)^{s+t+1}$, we have
\begin{align}
\E[Z(f_{i_r})] &= \sum_{s=0}^\infty
  \sum_{t=0}^\infty \E[Z(f_{i_r})\,
    |\, \calE_{s,t}]\Pr(\calE_{s,t})\\
&\geq q^3\sum_{s=0}^\infty \sum_{t=0}^\infty\sum_{k=0}^\infty
  \frac{1}{s+t+k+1}(1-q)^{s+t+k+1} \notag\\
&= q^3\sum_{\ell = 0}^\infty
  \binom{\ell+2}{2}\frac{1}{\ell+1}(1-q)^{\ell+1}
= \frac{q^3}{2}\sum_{\ell = 0}^\infty
  (\ell +2) (1-q)^{\ell+1} \notag\\
&\overset{(\star)}{=} \frac{q^3}{2} \frac{(1-q)(1+q)}{q^2}
  = \frac{1}{2}q(1-q^2),\label{eq:finish}
\end{align}
where equality~$(\star)$ is obtained by setting
$x=1-q$ in
\begin{equation*}
\sum_{\ell = 0}^\infty (\ell + 2) x^{\ell + 1}
  = \frac{d}{dx} \left(\sum_{\ell = 0}^\infty x^{\ell + 2}\right)
  = \frac{d}{dx} \frac{x^2}{1-x} = \frac{x(2-x)}{(1-x)^2}.
\end{equation*}
Finally, $q=1/\sqrt{3}$ is chosen to maximize
$q(1-q^2)/2$ among all values in $[0,1]$, and implies
by~\eqref{eq:finish},
\begin{equation*}
\E[Z(f_{i_r})] \geq \frac{1}{3\sqrt{3}},
\end{equation*}
thus proving~\eqref{eq:toProveZ}.
\end{proof}

\section{Conclusions}

We presented a $4$-competitive algorithm for the
free order model, which is a relaxed
version of the classical matroid secretary problem.
To the best of our knowledge, this is the first 
$O(1)$-competitive algorithm of a variant
of the matroid secretary problem with adversarial weight
assignments. The central question of whether
there is a $O(1)$-competitive algorithm for the classical
matroid secretary problem remains open.

Furthermore, a new approach to design
$O(1)$-competitive algorithms for the classical version of
the matroid secretary problem restricted to laminar
matroids was presented. For this special case,
a $16000/3$-competitive algorithm has
been found only very recently, using a rather 
involved method and analysis. Whereas relatively 
elegant and simple $O(1)$-competitive procedures have been
known for a variety of special cases 
of the matroid secretary problem, the  
$O(1)$-competitive algorithm for the laminar case was 
one of the most sophisticated procedures.
Our approach leads to simpler procedures with
considerably better competitiveness.
Furthermore, contrary to the previous approach for laminar
matroid~\cite{im_2011_secretary} and the very recent 9.6-competitive
algorithm~\cite{ma_2013_simulatedSTACS}, our algorithm is
order-oblivious,
and therefore implies a constant-competitive algorithm for
single-sample prophet inequalities as
shown in~\cite{azar_2014_prophet}.
A straightforward application of our high-level idea
already leads to a competitiveness of $27e/2$.
Additionally, we presented an improved version of the
algorithm and its analysis to obtain a 
$3\sqrt{3} e\approx 14.12$-competitive algorithm.

\bibliographystyle{plain}
\bibliography{lit}

\end{document}